\documentclass[12pt, letterpaper]{amsart} 
\usepackage{amsmath}         
\usepackage{amsfonts} 
\usepackage{graphicx} 
\usepackage{mathrsfs}
\usepackage{float}
\usepackage{bm}
\usepackage{amsmath,amstext,amsthm,amssymb,amsxtra, verbatim}
\usepackage[top=1.5in, bottom=1.5in, left=1.25in, right=1.25in]    {geometry}

\numberwithin{equation}{section}
\usepackage{tikz}
\usetikzlibrary{arrows}
\usepackage{mathtools}

\usepackage{amssymb}

\newtheorem{theorem}{Theorem}[section]
\newtheorem{lemma}[theorem]{Lemma}
\newtheorem{prop}[theorem]{Proposition}
\newtheorem{corollary}[theorem]{Corollary}
\theoremstyle{definition}
\newtheorem{definition}[theorem]{Definition} 
\newtheorem{exmp}[theorem]{Example}
\numberwithin{equation}{theorem}

\newcommand{\zz}{\mathbb{Z}}

\newcommand{\rr}{\mathbb{R}}
\newcommand{\cc}{\mathbb{C}}

\newcommand{\Flat}{\mathrm{Flat}} 
\newcommand{\rank}{\mathrm{rank}}

\title[Large Models]{Identifiability of Large Phylogenetic Mixtures for Many Phylogenetic Model Structures}
\author{Bryson Kagy and Seth Sullivant}
\email{bgkagy@ncsu.edu, smsulli2@ncsu.edu}
\address{ }

\begin{document}

\maketitle

\begin{abstract}

Identifiability of phylogenetic models is a necessary condition
to ensure that the model parameters can be uniquely determined from data.
Mixture models are phylogenetic models where the probability
distributions in the model are convex combinations of distributions
in simpler phylogenetic models.  Mixture models are used to model
heterogeneity in the substitution process in DNA sequences.  While many basic phylogenetic models are known to
be identifiable, mixture models in generality have only been
shown to be identifiable in certain cases. 
We expand 
the main theorem of [Rhodes, Sullivant 2012] to prove identifiability
of mixture models in equivariant phylogenetic models,
specifically the Jukes-Cantor, Kimura 2-parameter model, Kimura 3-parameter model  and the Strand Symmetric model. 
\end{abstract}
\makeatletter
\@setabstract
\makeatother

\section{Introduction} 
Phylogenetic models represent the evolutionary relationship among a collection
of taxa (short for taxonomic unit, which might be genes, species, or some other level of biological classification). 
Basic phylogenetic models are constructed by starting with a rooted tree where the taxa are vertices.
Taxa closer to the root correspond to more ancestral units in the evolutionary history. Associated to each edge in the tree is a transition matrix.  Each transition matrix encodes the  substitution probabilities
of the resulting 
characters along that edge.   Specifying a substitution model  amounts
to placing restrictions on the structure of the transition matrices.

In the General Markov model  no assumptions are made on the entries of the transition matrices,
whereas in the Jukes-Cantor model all off diagonal entries are equal so the transition matrices are highly constrained.  Models like the Kimura 2-parameter, Kimura 3-paramater, and the strand symmetric model sit 
in between these two extremes.

When using any statistical  model in practice,  it is desirable if a distribution arising from the model
uniquely determine the parameters that produced in, 
i.e.~the model should be \emph{identifiable}.
For basic phylogenetic models, these parameters are the tree parameter 
and the transition matrices.  The tree parameter is a discrete
parameter, in that there are only finitely many different
trees on $n$ leaves.  On the other hand, the space of possible
transition matrices forms a continuous parameter space.  Questions
of identifiability concern both the uniqueness of the discrete
parameters and the continuous parameters as determined from 
probability distributions produced by the model.
In this paper, by identifiable  we mean \emph{generically identifiable}, that is, the parameters that produce a distribution are identifiable
except possibly for a low dimensional subset of the parameter space. 
All the basic phylogenetic models are known to be generically identifiable
by classic results in the literature \cite{CHANG199651}.

This paper concerns the identifiability of mixture models, which are more complex models
which use the basic phylogenetic models on trees as building blocks.  
In a mixture model, we do not assume that all characters evolve according to the same underlying
phylogenetic tree with a fixed set of parameters.  Rather there are (hidden) classes, and depending on which class a character belongs to,
it evolves under the (basic) phylogenetic model according to that class.  Classes might correspond to
sites in DNA sequences that evolve at different rates (fast or slow).  Across larger regions of DNA, classes
could correspond to different genes which evolve according to different trees altogether.  
Mixture models have been used in situations where the precise nature
of heterogeneity across sites in the data is unclear \cite{Huelsenbeck2007}. 
Also, general mixture models
often contain phylogenetic network models as submodels.  For example, the displayed
tree model on a network (also called the network Markov model), is a submodel of a 
related mixture model, which can be useful for proving identifiability results about
phylogenetic network models, as in \cite{Englander2025}.

As mixture models have many more parameters than the basic phylogenetic models, questions of identifiability
become increasingly more difficult to answer and their analysis requires more advanced tools.
Limited work has been done on the identifiability  in group-based mixture models for specific numbers of classes and for completely general
tree structures.  For example, identifiability is known for 
2 class mixtures for the Jukes-Cantor and Kimura 2-parameter models \cite{Allman2011} for the Cavender-Farris-Neyman model and the Kimura 3-parameter model \cite{Hollering2021}.   For 3 class mixtures identifiability is only known for the Jukes-Cantor model \cite{Long2015}.  
 Work has also been done on the identifiability  of other types of mixture models such as Profile Mixture Models \cite{Yourdkhani2021},
 which represent a useful submodel of general mixture models.

Probably the strongest result on identifiability of mixture models,
the main theorem in \cite{Rhodes2012} gives conditions under which mixture models under the General Markov model are identifiable.   
While this result does not allow for arbitrary mixtures with completely
unrelated sets of trees, 
it does cover a large class of sets of trees, including the same tree mixtures which are the mostly commonly occurring.  In most settings,
these results allow for identifiability in a number of mixture classes
well above any that could be used in practice.  

All the results mentioned so far consider generic
identifiability, which, while a necessary consideration to use for mixture models, does not typically  produce
identifiability results for submodels.  Hence, the main results of \cite{Rhodes2012} only apply to the General Markov model,
and not to any of the submodels: Jukes-Cantor (JC), Kimura 2 parameter model (K2P), Kimura 3 parameter model (K3P) and the Strand Symmetric model (SSM). These submodels are called \emph{equivariant} because they have an underlying model symmetry to the assumptions they make on the transition matrices.
The goal of this paper is to broaden the tools developed from the main theorem in \cite{Rhodes2012} to 
also deduce the identifiability of those important submodels. 

The outline of this paper is as follows.
Before continuing onto the rest of the paper, in the next section we give some important definitions necessary to state the main results of this paper.  Theorem \ref{thm:main} gives a general framework for proving identifiability
results for phylogenetic mixture models and Corollary \ref{cor:JCetc}
establishes that the JC, K2P, K3P, and SSM all satisfy those conditions.
In Section \ref{trees and splits}, background on phylogenetic trees and the Markov, group-based, and equivariant models is given.   In Section \ref{Fourier}, the discrete Fourier transform is introduced, which will be useful to determine the rank of certain
matrices in our computations.
In Section \ref{tensor}, Kruskal's theorem  is introduced which gives bounds on the rank of tensors and will be crucial in the proof of the main theorem.
The technical heart of the paper is Sections \ref{sec:maintheorem}, \ref{sec:rankproperty},
and \ref{sec:noshuffle}.  Specifically, in Section \ref{sec:maintheorem},
we prove our general identifiability result showing that if a Markov model $\mu$
satisfies four key properties, then the associated mixture models will
be generically identifiable.  Sections \ref{sec:rankproperty} and \ref{sec:noshuffle}
then prove that the  JC, K2P, K3P, and SSM models all satisfy these technical conditions
to derive our main result.


\section{Key Definitions and Statement of Main Results}

The goal of this section is to give the necessary definitions
so that we can state our main results.  
This section just covers definitions to help make sense of the results. Further definitions will appear in subsequent sections.

 Let $T$ be a single phylogenetic tree with $n$ leaves.  Let $\mu$ be a 
 \emph{Markov model}, (for example, JC, K2P, K3P, SSM, or GMM standing for
 Jukes-Cantor, Kimura 2-parameter, Kimura 3-parameter, Strand Symmetric model, or General Markov Model, respectively).
 Let $S^{\mu}_{T}$ be the parameter space of continuous parameters for the underlying model $\mu$ on tree $T$ (which consists of transition matrices
 for each edge, and a root distribution). 
 Let $\kappa$ be the number of states of the random variables being considered ($\kappa = 4$ is most common in this paper).
 Once we specify these parts of a model, 
 there exists a parametrization $\psi^\mu_{T}$ which  gives the joint distribution of states on the leaves of $T$ as a function of continuous parameters which specify the root distribution of the tree and all of the transition matrices on its edges.  That is, 
  \[
  \psi^{\mu}_{T} : S^{\mu}_{T}  \to \Delta^{\kappa^n-1} 
  \]
  where $\Delta^{\kappa^n-1} \subset [0,1]^{\kappa^n-1}$ is the probability simplex of non-negative real vectors summing to $1$. Then the image of $\psi^{\mu}_{T}$ is the phylogenetic model $M^{\mu}_{T}$.

Given a multiset of $r$ phylogenetic trees $T_1, \ldots, T_r$, a phylogenetic mixture model can be defined as the set of  convex combinations of distributions from $M^{\mu}_{T_1}$, $\ldots$, $M^{\mu}_{T_r}$.

\begin{definition} \label{mixmodel}
Let $\mathbf{T}= (T_1, \dots T_r)$ be a $r$-tuple of trees each with $n$ leaves.  
Let $S^{\mu}_{\mathbf{T}}= S^{\mu}_{T_1} \times \dots S^{\mu}_{T_r} \times \Delta^{r-1}$, each with underlying model $\mu$. 
Then, the parametrization map of the \emph{mixture model} is defined by
\[
\psi^{\mu}_{\mathbf{T}} : S^{\mu}_{\bf{T}} \to \Delta^{\kappa^n-1} 
\]
where 
\[
\psi^{\mu}_{\mathbf{T}}(s_1, \dots, s_r, \pi) = \pi_1\psi^{\mu}_{T_1}(s_1) + \dots + \pi_r\psi^{\mu}_{T_r}(s_r)
\] 
with $s_i \in S_{T_i}$ , and $\pi \in \Delta^{r-1}$ is the vector of mixing parameters. 
 \end{definition}

Since there is no order on the combination of trees in mixture model, identifiability on mixture model will have to be defined up to reordering of the trees.
 
\begin{definition}
The tree parameters of an $r$-tree mixture model on $ \mathbf{T} =(T_1, \dots, T_r)$ are generically identifiable if for  generic choices of 
$s_1, \dots, s_r$ and $\pi$ with $s_i \in S^\mu_{T_i}$ , and $\pi \in \Delta^{r-1}$,
\[
\psi^{\mu}_{\bf{T}}(s_1, \dots, s_r, \pi) = \psi^{\mu}_{\bf{T}'}(s'_1, \dots, s'_r, \pi') 
\] 
implies that there is a permutation $\sigma \in \mathfrak{S}_r$ (the Symmetric group on $r$ letters)  
such that $\sigma \cdot \bf{T} = \bf{T}'$. 
\end{definition}

\begin{definition}
The continuous parameters of an $r$-tree mixture model on $ \mathbf{T} =(T_1, \dots, T_r)$ are \emph{generically identifiable} if for  generic choices of 
$s_1, \dots, s_r$ and $\pi$ with 
$s_i \in S^{\mu}_{T_i}$ , and $\pi \in \Delta^{r-1}$,
\[
\psi^{\mu}_{\mathbf{T}}(s_1, \dots, s_r, \pi) = \psi^{\mu}_{\mathbf{T}}(s'_1, \dots, s'_r, \pi') 
\] 
implies that there is a permutation $\sigma \in S_r$ such that $\sigma \cdot \mathbf{T} = \mathbf{T}, s'_i = s_{\sigma(i)}$ , and $\pi'_i = \pi_{\sigma(i)}$ for $i\in [r]$.
\end{definition}

In the identifiability of the continuous parameters of a phylogenetic
mixture model, we allow for label swapping, that permutes
parameters between two classes of the mixture model
that have the same underlying tree.  

We must define the class
of mixtures that we prove identifiability in.
This amounts restricting to a subset of $r$-tuples of trees.

\begin{definition}
    Let $\mathcal{T}(r,n,k)$ be the set of of $r$-tuples of $n$-leaf
    binary trees $(T_1, \dots, T_r)$ such that there exists
    a tripartition of the leaves $A|B|C$  and for each $i \in [r]$
    a vertex $v_i$ in tree $T_i$ such that the induced triparition of the
    leaves in $T_i$ induced by $v_i$ is $A|B|C$.  Furthermore, we
    assume that $\#A \geq \#B \geq \#C$ and $\#B \geq k$.    
\end{definition}

We now state our main result, which depends on some technical 
definitions which will appear in later sections.
However, this gives us the flavor of the results that can be achieved
with these methods, and the fact that certain specific
properties need to be proved for a model $\mu$ to deduce
identifiability for a mixture model.  These properties are
\emph{standard},  \emph{the rank property}, \emph{the extended rank property},
and \emph{the No Shuffling Property}.

 \begin{theorem}  \label{thm:main}
    Suppose that a phylogenetic model $\mu$ on $\kappa> 2$ states
        satisfies the following properties:
    \begin{enumerate}
    \item  $\mu$ is standard,
    \item  $\mu$ has the rank property $RP(r,k)$
    \item  $\mu$ has the extended rank property $ERP(r,k)$
         \item $\mu$ has the No Shuffling Property
         \item  $r \leq \kappa^{k -1}$, and 
         \item  $n \geq 2k + 1$.
    \end{enumerate}
     Then both tree parameters and the numerical parameters of $M^{\mu}_{\bf{T}  }$ are generically identifiable
     in the class of trees $\mathcal{T}(r,n,k)$.
 \end{theorem}

 \begin{corollary}\label{cor:JCetc}
     Suppose that $\mu$ is the Jukes-Cantor model, Kimura-2 or 3-parameter model, 
     Strand symmetric model, or general Markov model.  Suppose that $n \geq 2k+1$
     and $r \leq 2^k - k$.  Then both tree parameters and the numerical parameters of $M^{\mu}_{\bf{T}  }$ are generically identifiable
     in the class of trees $\mathcal{T}(r,n,k)$.
 \end{corollary}

 Part of our goal in these proofs is to try to give a uniform argument
 to deduce identifiability for these models.  A key component of the
 proof of Corolllary \ref{cor:JCetc} 
 is to prove that $\mu$ has the rank property $RP(r,k)$ and extended
 rank property $ERP(r,k)$ for the specific values of $r$ and $k$.   
To show this we prove both properties it for the Jukes-Cantor
 model for the particular values of $r$ and $k$ in the Theorem.  Since 
 Jukes-Cantor transition matrices appear in all the other models, this
 implies the rank property for the larger models for those
 values of $r$ and $k$.  It
 is probably possible to derive stronger results for the other models
 by more carefully analyzing the rank property and extended rank property in other
 models.  
 We discuss this in more detail in Section \ref{sec:rankproperty}.


 \section{Markov models on trees}\label{trees and splits}

In this section, we review background on the combinatorial structure
of trees, including splits and tripartitions associated to tree
edges and vertices.  We introduce the Markov models on
a tree that are the main object of study in this work.

A tree $T = (V, E)$ is a graph  that is connected and has no cycles. 
The degree of a vertex is the number of edges incident to the vertex.  A vertex of degree 1 is called a leaf.  
Phylogenetic trees  represent the evolution of taxa with the vertices each representing a different taxa. The closer a vertex is to the root vertex the older that taxa  is evolutionarily.  The leaves of the tree will represent taxa that are alive today. 
This gives the leaves special importance since the DNA of those species is observable. The following definitions of phylogenetic trees are standard (see \cite{Steel2016}).

 \begin{definition}
Let $T$ be a tree and $X$ be a set of labels.   Let $\phi$ be a map from $X$ to the vertices of $T$.
The pair $(T, \phi)$ is called a \emph{phylogenetic $X$-tree} if the image of $\phi$ is exactly the set of leaves of $T$ and $\phi$ is injective.

A \emph{binary phylogenetic X-tree} is a phylogenetic $X$-tree where all non-leaf vertices have degree 3.

A \emph{rooted tree} is a tree with a distinguished vertex called the root.  
\end{definition}

Removing any edge of a phylogenetic tree separates the leaf labels into two parts,
based on which connected component the leaf is in.  This leads to following definition.

  \begin{definition} \label{part}
A \emph{split} $A|B$ of $X$ is a partition of $X$ into two nonempty disjoint sets. 
A split is considered \emph{valid} for an $X$-tree $T$ if it is obtained by removing an edge from $T$ and taking $A$ and $B$ to be the two sets of leaf labels for the two component graphs.   

Similarly, a tripartition $A|B|C$ is obtained from a binary tree by removing an internal vertex and $A$ and $B$ and $C$  being the three sets of leaf labels for the three component graphs.
   \end{definition} 

 Let $T$ be a phylogenetic tree.  We associate a random variable $X_v$ to each
 vertex in $T$.  Each of the random variables $X_v$ has $\kappa$ states for
 a fixed value $\kappa$. In general terms, the state space of $X_v$ is $[\kappa] = \{1,2, \ldots \kappa\}$.   
  If $\kappa=4$ this alphabet can be thought of as DNA bases $\{A,C,G,T\}$. There are other values for $\kappa$ that are relevant biologically, such as  $\kappa=20$ for amino acids,  $\kappa=61$ for DNA codons, and  $\kappa=2$ for DNA where purines and pyrimidines are grouped together. This paper will be primarily focusing on $\kappa=4$. 
  
Before defining 
Markov models on a tree, we will motivate their definition.  
Consider for each vertex a sequence using the characters of the alphabet. 
If $\kappa=4$ that sequence on the vertex would be the DNA sequence for a gene and each vertex would represent a DNA sequence for a specific taxa. 
For simplicity, only point substitutions will be considered, i.e.~mutations that change one DNA letter to another. Thus each sequence will be of the same length for all
vertices in the graph, but possibly have different characters in each location. 
The Markov model on the tree is a model for how these sequences change over time.
If we assume that each site in the sequence evolves independently
under the same process, we can describe the model just by
assuming we have sequences of length $1$, or a single character.  
 Thus, for each vertex we will associate a random variable with $\kappa=4$ states where the internal non-leaf vertices will be hidden random variables (since they
 are unobserved).  

To each edge in the graph $i \rightarrow j$ is associated a 
Markov transition matrix $M^{ij}$, which is a $\kappa \times \kappa$
matrix whose $(x_i, x_j)$ entry is the conditional probability 
$P(X_j = x_j | X_i = x_i)$.  That is:
\[
M^{ij}(x_i, x_j)  = P(X_j = x_j | X_i = x_i).
\]
To be a Markov transition matrix $M^{ij}$, we require that all entries
are nonnegative, and that the row sums are all equal to one.  

Once a Markov transition matrix is specified for each edge, and a 
distribution of states $\pi \in \Delta_{r-1}$ is specified, we can write down
the joint distribution of all states in the model.  Let $p(x,y)$
denote the joint probability distribution of all random variables,
where $x$ corresponds to leaves and $y$ corresponds to
interior vertices.  
Let $\rho$ denote the root vertex.
Then
\[
p(x, y)  =  \pi(y_\rho) \prod_{i \rightarrow j \in E(T)}  M^{ij}(x_i,x_j).
\]
Since we are typically only interested in the distribution of states
at the leaves (because we do not have access to data from
the taxa at interior nodes), then we have the probability distribution
of interest is
\[
p(x)  =  \sum_{y \in [\kappa]^{\mathrm{Int}(T)}} p(x,y)
\]
where $\mathrm{Int}(T)$ denotes the interior vertices of $T$.

 \begin{figure}[t] \label{tree}

\begin{center}

\begin{tikzpicture} 
 
\node {a}
    child  {node {1} edge from parent }
    child {node  {b} 
    child {node {2}}
    child {node {3}} 
    edge from parent  node [right, brown] {}};
 
\end{tikzpicture}

\[
M^{b2} = 
\begin{pmatrix} 
	0.25 & 0.3 & 0.15 & 0.3 \\
	0.3 & 0.25 & 0.4 & 0.05\\
	0.15 & 0.4 & 0.25 & 0.2 \\
	0.3 & 0.05 & 0.2 & 0.45 \\
	\end{pmatrix} 
    \]

\caption{A phylogenetic model with 3 leaves and two internal vertices  a and b and a transition matrix}
\end{center}
 \end{figure}
 
\begin{exmp}
 Consider the rooted tree in Figure \ref{tree}.
 Let the transition matrix written here be associated to the edge from $b  \to 2$. 
 Label the columns and rows of the matrix with $A, C, G, T$.
 Then the entries in this transition matrix represent probabilities of DNA substituion
 when going from the species represented in the vertex b to the one represented with the vertex 2. Explicitly the entry $M^{b2}(x_b,x_2)$ 
 is the probability that the variable at $X_b$ mutates from the state $x_b$  
 to the state $x_2$ at vertex $2$.   
 So for example there is a $40\%$ chance $X_b = C$ changes to  $X_2 = G$.
\end{exmp}

Usually when defining a phylogenetic model, we put restrictions on
the structure of transition matrices that are used in the model.  
We call such restriction on the transition matrices a \emph{Markov model $\mu$}.
Examples of Markov models are the Jukes-Cantor model, Kimura 2 and 3 parameter models,
the strand symmetric model, and the general Markov model.
The resulting structures on the transition matrices are as follows.
The matrix $M^{K3P}$ denotes a generic matrix for the Kimura $3$-parameter model:
\[
M^{K3P} = 
\begin{pmatrix}
\alpha & \beta & \gamma & \delta \\
\beta & \alpha & \delta  & \gamma \\
\gamma & \delta  & \alpha & \beta \\
\delta  & \gamma & \beta & \alpha
\end{pmatrix}.
\]
Note that this matrix only has $3$ free parameters, because we require that 
$\alpha + \beta + \gamma + \delta = 1$.  
The Kimura $2$-parameter model arises as a submodel obtained by requiring
that $\beta =  \delta$.  And the Jukes-Cantor model is a further submodel where 
$\beta = \gamma = \delta$.

The Strand Symmetric Model is a larger model where the generic transition matrices
have the following form:
\[
M^{SSM} = 
\begin{pmatrix}
\alpha & \beta & \gamma & \delta \\
\epsilon & \zeta & \eta  & \theta \\
\theta & \eta  & \zeta & \epsilon \\
\delta  & \gamma & \beta & \alpha
\end{pmatrix}.
\]    
This model has $6$ free parameters associated to each transition matrix (again, because
the row sums must equal $1$).  
Note that this model contains the JC, K2P, and K3P models as submodels.  
Finally, the general Markov model makes no restrictions at all on the transition
matrices, except that the matrices are transition matrices.    

All of these models are examples of what are called \emph{equivariant models},
a useful class of phylogenetic models with key symmetry properties
that makes them easier to analyze \cite{Casanellas2023}.  A key family of which are the 
group-based models, of which the JC, K2P, and K3P are special cases.

\begin{definition}
A phylogenetic model $\mu$ is \emph{group-based} if for all transition matrices, $M$, associated to its edges there exits a function $f: G \to \mathbb{R}$ such that $M(g,h) = f(g-h)$ where $G$ is an abelian group and $M(g,h)$ denotes the probability of going from state $h$ given that the variable is in state $g$ in the previous vertex. 
 \end{definition}

In the next section, we explain how the discrete Fourier transform
can simplify the presentation of the group-based models.


\section{The Discrete Fourier Transform}\label{Fourier}

The discrete Fourier transform is a useful tool that,
for certain phylogenetic models, can be applied to simplify the
parametrization of the model, to get a parametrization in a new
coordinate system that involves polynomials with fewer terms.
For the case of group-based models (including the Jukes-Cantor,
Kimura 2 and 3 parameter models), it simplifies the phylogenetic
model on a tree into a toric variety.  For the strand symmetric model,
the discrete Fourier transform does not produce a toric variety,
but another variety that is considerably simpler than the standard
parametrization.

\begin{definition}\label{def:Fourier}
Let $G$ be a finite abelian group and let $\hat{G} = \mathrm{Hom}(G, \cc^\times)$ 
be its dual group of one dimensional representations.   
Let $f$ be a function $f : G \to \mathbb{C}$.
The \emph{discrete Fourier transform} of $f$ is the function 
$\hat{f} : \hat{G} \to \mathbb{C}$ defined by 
\[ \hat{f}(\chi) = \sum_{g \in G} f(g)\chi(g). \]
\end{definition}

Definition \ref{def:Fourier} provides the Fourier 
transform for an arbitrary finite abelian group.
We will see how this is applied and can simplify the
parametrization for a group-based model and other equivariant models.

First we describe the transformation for the group-based models,
and then we show how the parametrization looks of the group-based models
in the Fourier coordinates. 
Let $T$ is a rooted tree with $m$ leaves, in a group based model,
for each edge we have a transition matrix $M^e$ that satisfies $M^e(g,h) = f^e(g-h)$ for some particular function $f^e : G \to \mathbb{R}$. 
Additionally there is a root distribution $\pi : G \to \mathbb{R}$.
For the group-based models we consider here, we will assume that the root
distribution $\pi$ is always the uniform distribution.

Let $H(e)$ be the head vertex of edge $e$ and $T(e)$ be its tail vertex. 
Let $\lambda(e)$ be the set of leaves below $e$ and let $Int$ be the set of interior vertices of $T$.
Then the joint probability of observing $(g_1, \dots , g_m)$ at the leaves is:
\begin{align*}
 p(g_1 , \dots , g_m) &= \sum_{g \in G^{Int}} \pi(g_r) \prod_{e \in T} M^e(g_{T(e)}, g_{H(e)}) \\
 &= \sum_{g \in G^{Int}} \pi(g_r) \prod_{e \in T} f^e(g_{T(e)}- g_{H(e)}).
\end{align*}
See Chapter 15 of \cite{Sullivant2018Book} for more details.  

Now we compute the Fourier transform  $p : G^m \to \mathbb{R}$
over the group $G^m$. Note that the dual group of a product group
is the product of the dual groups.  So, as a linear change of coordinates, this transformation is
\[
\hat{p}(\chi_1, \ldots, \chi_m) =  \sum_{g_1, \ldots, g_m \in G} \chi_1(g_1) \cdots
\chi_m(g_m)  p(g_1, \ldots, g_m).
\]

The Fourier transform of $p$ can be described with the following
product formula:
\[ 
\hat{p}(\chi_1 , \dots , \chi_m)= 
  \hat{\pi}\left(\prod_{i=1}^m \chi_i\right)\prod_{e \in T} \hat{f}^e \left(\prod_{i\in \lambda(e)} \chi_i\right)
\]
This was first discussed in \cite{Evans1993} and \cite{Hendy1996}.
 
\begin{exmp}
Consider the rooted 3 leaf tree with underlying Jukes-Cantor model. Let $f^i(A)=b_i$ and $f^i(C)=f^i(G)=f^i(T)=a_i$.  We make the identification that $A = (0,0)$, $C = (0,1)$, $G = (1,0)$, and $T = (1,1)$.
The four characters of $\mathbb{Z}_2\times \mathbb{Z}_2$ are the rows of the following character table:

\begin{table}[H]
\centering
\begin{tabular}{r|rrrr}
         & A & C & G & T \\
         \hline
$\chi_A$ & $1 $    & $1 $    & $1 $     & $1 $     \\
$\chi_C$ & $1 $     & $-1$    & $1 $     & $-1$     \\
$\chi_G$ & $1 $     & $1 $     & $-1$     & $-1$    \\
$\chi_T$ & $1 $     & $-1$     & $-1$     & $1 $    
\end{tabular}
\end{table}
So the Fourier transform for a particular edge is
\[
\hat{f}^i(\chi_A) =  3a_i + b_i \quad \quad \mbox{ and } \quad \quad \hat{f}^i(\chi_C) = \hat{f}^i(\chi_G)
= \hat{f}^i(\chi_T) = b_i - a_i.
\]
Note that the Fourier transform of the uniform distribution has 
\[
\hat{\pi}(\chi_A) = 1, \quad \quad \mbox{ and } \quad \quad \hat{\pi}^i(\chi_C) = \hat{\pi}^i(\chi_G)
= \hat{\pi}^i(\chi_T)  = 0
\]
Let $x_i=3a_i+b_i$ and $y_i=b_i- a_i$.

We use  the shorthand $A = \chi_A, \ldots,$ and we write $q_{AAA} = \hat{p}(\chi_A, \chi_A, \chi_A), \ldots.$  Then we have
\[
q_{AAA} = x_1 x_2 x_3 , \quad   q_{ACC} = q_{AGG} = q_{ATT} =   x_1 y_2 y_3, 
\]
\[
 q_{CAC} = q_{GAG} = q_{TAT} =   y_1 x_2 y_3, \quad q_{CCA} = q_{GGA} = q_{TTA} =   y_1 y_2 x_3,
\]
\[
q_{CGT} = q_{CTG} = q_{GCT} = q_{GTC} = q_{TCG} = q_{TGC} =  y_1 y_2 y_3,
\]
\[
q_{\chi_1 \chi_2 \chi_3} = 0  \quad \mbox{ otherwise. }
\]
\end{exmp}

For a group-based model on the group $G = \zz_2 \times \zz_2$, we can express
the parametrization in the following form.  We have Fourier parameters
$a^{A|B}_{g}$ for each split $A|B$ in the tree, and each group element $g$.
The parametrization in Fourier coordinates looks like
\[
q_{g_1 \cdots g_n}  =  \begin{cases}
\prod_{A|B \in \Sigma(T)}  a^{A|B}_{\sum_{a \in A} g_a}  &  \mbox{ if }  \sum_{i = 1}^n g_i = (0,0)  \\
0 & \mbox{ otherwise.}
\end{cases}
\]
Then, depending on the specific model, there might be equalities that hold between
some the Fourier parameters. In the Jukes-Cantor model,
 for each split $A|B \in \Sigma(T)$, we have that
\[
a^{A|B}_{(0,1)} = a^{A|B}_{(1,0)} = a^{A|B}_{(1,1)} .
\]


\section{Tensors and Flattenings}\label{tensor}

In this section, we discuss how to think about the probability distribution
associated to a phylogenetic mixture model as a multi-way tensor.
The tensor perspective is useful, because we can also consider flattenings
of our tensor to lower order tensors, and matrices, in order to use Kruskal's theorem and matrix ranks of flattenings to prove identifiability
results. The  approach us using Kruskal's theorem to 
prove identifiability results first appeared in \cite{Allman2009}.

For a tree with $n$ leaves, and random variables with $\kappa$ states,
the joint probability distribution in our model is naturally considered
as a $\kappa \times \dots  \times \kappa$ $n$-tensor, inside of $\bigotimes_{i = 1}^n  \rr^\kappa$.  
From this $n$-way tensor, we can
also consider the tensor naturally as a lower order tensor, by
grouping indices according to a set partition of the coordinates.
Such a reorganization of the tensor is called a \emph{flattening}.

 \begin{definition}
 Let $P$ be an order $n$ tensor in $\bigotimes_{i =1}^n \mathbb{R}^{\kappa_i}$.
 Let $A_1| \dots |A_k$ be a partition of $[n]$.   
 The flattening of $P$ according to $A_1| \dots |A_k$, is an order-$k$ tensor
 obtained by grouping indices according to the partition $A_1| \dots |A_k$,
 and is denoted $\mathrm{Flat}_{A_1| \cdots | A_k}(P)$.
 Specifically, $\mathrm{Flat}_{A_1| \cdots | A_k}(P)$ is the image of $P$
 under the natural map 
 \[
 \bigotimes_{i =1}^n \mathbb{R}^{\kappa_i}   \rightarrow
 \bigotimes_{i = 1}^k  \left(  \bigotimes_{a \in A_i}  \rr^{a}  \right).
 \]
 In particular, $\mathrm{Flat}_{A_1| \cdots | A_k}(P)$ is a 
 $\prod_{a_1 \in A_1} \kappa_{a_1} 
 \times  \cdots \times  \prod_{a_k \in A_1} \kappa_{a_k}  $ tensor.
 \end{definition}

Note that in the specific case when $\kappa_i = \kappa$ for all $i$, and  if $k=2$, then 
we have our partition $A|B$ and $\mathrm{Flat}_{A|B}(P)$ is a matrix of
size $\kappa^{\#A} \times \kappa^{\#B}$.

 \begin{exmp}
Let $T$ be a $3\times 3 \times 3$ tensor with generic entries $t_{ijk}$,
and consider the partition $A|B = 1|23$.
Then 
\[
\Flat_{1|23}=\begin{pmatrix}
t_{111} & t_{112} & t_{113} & t_{121} & t_{122} & t_{123} & t_{131} & t_{132} & t_{133} \\
t_{211} & t_{212} & t_{213} & t_{221} & t_{222} & t_{223} & t_{231} & t_{232} & t_{233} \\
t_{311} & t_{312} & t_{313} & t_{321} & t_{322} & t_{323} & t_{331} & t_{332} & t_{333} \\
\end{pmatrix} 
\]
\end{exmp}

One key advantage to considering out probability distributions
as tensors is that we can use tensor rank as a tool 
to get at questions of identifiability in phylogenetic models.
To explain the key tool in this area, Kruskal's theorem,
we need to review some concepts related to tensor rank.

For $j = 1,2, 3$, let $\bm{m}^j = (m^j(1), \ldots, m^j(\kappa_j))  \in \rr^{\kappa_j}$.    Let  $\bm{m}^1 \otimes \bm{m}^2 \otimes \bm{m}^3$
denote the $\kappa_1 \times \kappa_2 \times \kappa_3$ tensor whose
$(u,v,w)$ entry is $m^1(u)m^2(v) m^3(v)$.  The tensor 
$\bm{m}^1 \otimes \bm{m}^2 \otimes \bm{m}^3$ is called a \emph{rank 1 tensor}.
A tensor is said to have rank $r$ if it can be written as the sum of $r$ rank $1$
tensors, and cannot be written as the sum of $r-1$ rank $1$
tensors.  Kruskal's theorem concerns the uniqueness of the representation
of a rank $r$ tensor as the sum of $r$ rank $1$ tensors.  
To express this, we first introduce the triple product notation
of matrices.

 \begin{definition}
 Let $M_j$ be a $r \times \kappa_j$ matrix with $i$th row $\bm{m_i^j} = (m_i^j(1), \dots m_i^j(\kappa_j)).$ Then let $[M_1, M_2, M_3]$ be defined by 
\[
[M_1, M_2, M_3] = \sum_{i=1}^r \bm{m}_i^1 \otimes \bm{m}_i^2 \otimes \bm{m}_i^3. 
\]
The triple product
$[M_1, M_2, M_3]$ is a $ \kappa_1 \times \kappa_2 \times \kappa_3$
tensor and its $(u,v,w)$ entry is 
   \[ 
   [M_1, M_2, M_3]_{(u,v,w)} = \sum_{i=1}^r m_i^1(u) \otimes m_i^2(v)  \otimes m_i^3(w) .
   \]
\end{definition}

Note that there is a natural nonuniqueness that is always present.
If $\Pi$ is an $r \times r$ permutation matrix, and $r \times r$ diagonal matrices
$D_1, D_2, D_3$ such that $D_1D_2D_3 = Id_r$, then
\[
[M_1, M_2, M_3]  =  [\Pi D_1 M_1, \Pi D_2 M_2,  \Pi D_3 M_3].
\]
If 
\[
[M_1, M_2, M_3] = [N_1, N_2, N_3]
\]
and there exist $\Pi$ and $D_1, D_2, D_3$ such that 
\[
[N_1, N_2, N_3] = [\Pi D_1 M_1, \Pi D_2 M_2,  \Pi D_3 M_3]
\]
when we say that $N_1, N_2, N_3$ are related  to $M_1, M_2, M_3$
by scaling and simultaneous permutation of the rows.  If this is the 
only way we can have $[M_1, M_2, M_3] = [N_1, N_2, N_3]$ then
we say that $[M_1, M_2, M_3]$ uniquely determines $M_1$, $ M_2$ , $M_3$ up to simultaneous permutation and scaling of the rows.
Note that this means that the rank one tensors that appear in the decomposition
are themselves unique.  

The existence of a  unique of the decomposition in the Kruskal theorem
depends on the notion of the Kruskal rank of a matrix.  

\begin{definition}
    Let $M$ be a $r \times \kappa$ matrix. 
    The \emph{Kruskal rank of a matrix $M$}, written as 
    $\rank_{K}(M)$ is the largest $k$ such that every 
    subset of $k$ rows of $M$ is linearly independent.
\end{definition}

Note that $\rank_{K}(M) \leq \rank(M) $ but if $M$ is of full row rank 
(that is, $\rank(M) = r$, it will also be full Kruskal rank.
Now we are ready to state Kruskal's Theorem.

 \begin{theorem} (Kruskal's Theorem) \cite{Kruskal77}
 For $j = 1,2,3$, let $M_j$ be an $r \times \kappa_j$ matrix and
 let $I_j= \rank_K M_j$. If 
 \[ 
 I_1 +I_2 +I_3 \geq 2r +2 
 \] 
 then $[M_1, M_2, M_3]$ uniquely determines $M_1$, $M_2$, and $M_3$ up to simultaneous permutation and scaling of the rows.  
 \end{theorem}

The motivation for considering tensor rank comes from the 
fact that the probability distribution from a single tree
can be represented as a triple product, when considered around
a given vertex.  Specifically, let $T$ be a tree, let $v$ be a 
trivalent vertex in the $T$.  We can assume that $v$ is the root
of the tree.  The vertex $v$ introduces a tripartition $A|B|C$
of the leaves.  Let $\mathrm{diag}(\pi)$ be a diagonal matrix
with the root distribution.  Let $M_A$ be the 
$\kappa \times \kappa^{\#A}$  matrix which is the flattenings
of the conditional distribution of $X_A$ given $X_v$.
Similarly, define $M_B$ and $\tilde{M}_C$ as conditional distributions of
$X_B$ given $X_v$ and $X_C$ given $X_v$, respectively.  Then,
if $P$ is the joint distribution of states at the leaves
of our tree $T$, we have
\[
\Flat_{A|B|C}(P)  =  [M_A, M_B, \mathrm{diag}(\pi) \tilde{M}_C]  =[M_A, M_B, M_C].
\]
where we let $M_C = \mathrm{diag}(\pi) \tilde{M}_C.$
This representation was first developed in
\cite{Allman2011}, and used to prove identifiability results
in a number of hidden variable models including phylogenetic
models.

Note that we could have absorbed the diagonal matrix 
$\mathrm{diag}(\pi)$ into any
of $M_A, M_B$ or $\tilde{M}_C$.  Since we are always working with
generic distributions, this would not change the Kruskal rank or
tensor rank.  Hence we can consider $M_A$, $M_B$, and $M_C$
as either conditional distributions given $X_v$, or just as joint distributions.


\section{A General Theorem on Identifiability of Phylogenetic Mixtures}
\label{sec:maintheorem}

Our goal in this section is to prove our main general result
about the identifiability of phylogenetic mixture models, Theorem \ref{thm:main}.
To do this, we must first define the rank property, extended rank property,
the standard Markov model property, and the No Shuffling Property, and
show how they are related to identifiability of mixture models.

\begin{definition}
    The Markov model $\mu$ is said to have the 
\emph{Rank property} $RP(r,n)$ if the following holds:
For any trees $T_1, \ldots, T_r$ on $n$ leaves, and generic probability
distributions $P_i \in M^\mu_{T_i}$, the matrix
\[
\begin{pmatrix} \Flat_{[n-1]|n}(P_1) & \Flat_{[n-1]|n}(P_2)  & \cdots & \Flat_{[n-1]|n}(P_r)
\end{pmatrix}
\]
has rank $r \kappa$.
\end{definition}

\begin{definition}
     The Markov model $\mu$ is said to have the 
\emph{extended rank property} $ERP(r,n)$ if the following holds:
For any tree $T_1$ on $n+1$ leaves that does not contain the split $[n-1]| \{n, n+1\}$,
and any trees $T_2, \ldots T_r$  on $n$ leaves, and generic probability
distributions $P_i \in M^\mu_{T_i}$, the matrix
\[
\begin{pmatrix} \Flat_{[n-1]|\{n, n+1\}}(P_1) & \Flat_{[n-1]|n}(P_2)  & \cdots & \Flat_{[n-1]|n}(P_r)
\end{pmatrix}
\]
has rank $(r-1) \kappa + \kappa^2$.
\end{definition}

A useful fact about both the rank property and extended rank property 
is that they satisfy persistence properties in their values.

\begin{prop}
    Let $\mu$ be a Markov model such that the rank property $RP(r,n)$ holds.
    Then the rank property $RP(s, m)$ holds for all $s \leq r$ and $m \geq n$.
    Similarly, if the extended rank property $ERP(r,n)$ holds, then
     the extended rank property $ERP(s, m)$ holds for all $s \leq r$ and $m \geq n$.
\end{prop}

\begin{proof}
    Suppose that $RP(r,n)$ holds.  Let 
    \[
Q = \begin{pmatrix} \Flat_{[n-1]|n}(P_1) & \Flat_{[n-1]|n}(P_2)  & \cdots & \Flat_{[n-1]|n}(P_r)
\end{pmatrix}.
\]
    Then the associated matrix $Q'$ which is used to check $RP(s,n)$
    for $s \leq r$ is obtained from $Q$ by taking a submatrix of the columns of $Q$. 
    Since
    $Q$ had full column rank $Q'$ must have full column rank as well.

    Now suppose that $m \geq n$.   Construct the matrix $Q'$ to check $RP(r,m)$.
    Each of the distributions obtained to form $Q$ is gotten by marginalizing
    the corresponding distribution from $Q'$.  We can assume we are marginalizing
    some collection of random variables besides the last variable, $m$.  On the level
    of matrices, $Q$ is obtained from $Q'$ by adding together rows of $Q'$.  
    Since we assumed that $Q$ has rank $r \kappa$ and $Q$ is obtained as
    a linear transformation of $Q'$, and has exactly $r \kappa$ columns, it
    must also have had rank $r \kappa$, so $\mu$ satisfies $RP(r,m)$.

    The same arguments work verbatim for the extended rank property.
\end{proof}

One further property we need of our phylogenetic models is a certain notion of
standardness of the set of transition matrices, as that will
allow us to deduce that the matrices have a suitable Kruskal rank.

\begin{definition}
    A Markov  model $\mu$ is \emph{standard} if it
    satisfies the following conditions:
  \begin{itemize}
      \item the model is generically identifiable on trees
      \item the model contains the identity matrix, and
      \item  the model contains
    two transition matrices $M_1$ and $M_2$ such that  any pair of
    column vectors obtained from taking two columns of $M_1$, or one column of $M_1$
    and one column of $M_2$ are linearly independent.  
  \end{itemize}  
 
\end{definition}

\begin{prop}
  Let $\mu$ be a standard Markov model.  Let $T_1, \ldots, T_r$ be
  trees on $n$ leaves with $n \geq 2$.  Let  $P_1, \ldots, P_r$ be generic 
  probability distributions $P_i \in M^\mu_{T_i}$. 
  Then the matrix
  \[
Q = \begin{pmatrix} \Flat_{[n-1]|n}(P_1) & \Flat_{[n-1]|n}(P_2)  & \cdots & \Flat_{[n-1]|n}(P_r)
\end{pmatrix}
  \]
  has Kruskal rank $\geq 2$ for any $r$. 
\end{prop}

\begin{proof}
    By marginalizing, we can assume that all the trees are just a single edge with $2$
    leafs.  To see that the matrices have Kruskal rank $\geq 2$, we just need to look
    at the case where $r = 1, 2$, and taking any pair of columns from the resulting
    two transition matrices.  However, the definition of standard specifically
    covers these cases, which shows that matrix $Q$  has Kruskal rank at least $2$.
\end{proof}

Last we come to the No Shuffling Property, which says that 
parts of probability distributions on different trees cannot
be combined to get a distribution on another tree.

\begin{definition}
A Markov model $\mu$ is said to the have the \emph{No Shuffling Property} 
if the following holds:
For any trees $T_1, \ldots, T_r$ with $n \geq 3$ leaves 
and generic probability distributions
$P_i \in M^\mu_{T_i}$, let $Q$ be the matrix
\[
Q = \begin{pmatrix} \pi_1 \Flat_{[n-1]|n}(P_1) & \pi_2 \Flat_{[n-1]|n}(P_2)  & \cdots & \pi_r\Flat_{[n-1]|n}(P_r)
\end{pmatrix}
\]
where $\pi \in \Delta_r$ is generic.  
Form a new matrix $Q'$ by taking $\kappa$ not necessarily distinct columns of $Q$, and let $P$
be the resulting tensor.  If $P \in M^\mu_T$ for some $T$, then all
columns of $Q'$ must have come from the same $P_i.$
\end{definition} 

The importance of the No Shuffling Property is that it implies that if
we permute the columns of $Q$, we can still tell which columns belong together
(coming from the same distribution and same tree). This follows because 
it is not possible
to get a distribution from a tree by taking a combination of some of the columns
from different trees.

Our goal at this point is to show how a model that satisfies
these properties can be made to follow the proof
strategy from \cite{Rhodes2012}, closely following the
outline of those ideas, but substituting in these new properties
in place of directly using the General Markov model to prove the results.

\begin{lemma} \label{lem:flatrank}
    Let $\mu$ be a  Markov model that satisfies the extended rank property $ERP(r,k)$ and let $r \leq \kappa^{k -1}$.
    Let $P$ be a probability that is a generic mixture of 
    of $r$ generic distributions from trees $T_1, \ldots, T_r$ on
    $n$ leaves.  Suppose that $A|B$ is a split with $\#A > k $ and $\#B > k$.
    \begin{enumerate}
        \item If every tree $T_1, \ldots, T_r$ displays the split
        $A|B$ then $\rank (\Flat_{A|B}(P)) \leq \kappa r.$
        \item  If there is some tree $T_i$ that does not display the
        split $A|B$ then $\rank (\Flat_{A|B}(P)) > \kappa r.$
    \end{enumerate}
\end{lemma}

\begin{proof}
         If $A|B$ is compatible with all trees in $\mathbf{T},$ then, by passing to binary resolutions of the $T_i,$ let  $A|B$ be associated to the edge $e_i = (a_i, b_i)$ in $T_i.$ Then one sees that
 \[
 \Flat_{A|B}(P) = M^T_A Q M_B. 
 \]
  Here $Q$ is the $\kappa r \times \kappa r$ block-diagonal matrix whose 
  $i$th $\kappa \times \kappa$ block gives the joint probability distribution of states for the random variables at $a_i$ and $b_i$, 
  weighted by the component proportion $\pi_i.$ 
  The matrices $M_A$, $M_B$ are stochastic, of sizes $ \kappa r \times \kappa^{\# A}$,  $\kappa r \times \kappa^{\# B}$, with entries in the $i$th block of $\kappa$ rows giving probabilities of states of variables in $A$, $B$ conditioned on states at $a_i$, $b_i$. 
  Since $Q$ is a $\kappa r \times \kappa r$  matrix it has rank at most $\kappa r$,
  which implies that  $\Flat_{A|B}(P)$ has rank at most $\kappa r$
  as well.  

 Suppose next that $A|B$ is not compatible with at least one of the trees in 
 $\mathbf{T},$ say $T_1.$ 
 To show that $\text{Flat}_{A|B}(P)$ generically has rank greater than 
 $\kappa r$, it is enough to give a single choice of parameters producing such a rank. Indeed, this follows from Proposition 3.2 in \cite{Rhodes2012}, 
 applied to the model and the variety of matrices of rank at most $\kappa r$. 
 
For each $T_i$ with $i>1$ choose all Markov matrices 
for all internal edges of $T_i$ to be the identity, $I_{\kappa}$. 
Since $T_1$ is not compatible with $A|B$, by Theorem 3.8.6 of [19], 
it has an edge $e= (c,d)$, with associated split $C|D$, 
such that all four sets $A\cap C$, $A\cap D$, $B\cap C$, $B\cap D$ are nonempty. 
For all internal edges of $T_1$ except $e$, 
choose Markov matrices to be $I_{\kappa}$ as well. 
Since the effect of an identity matrix on an edge is the same as contracting that edge. For simplicity, we will refer to these contradicted trees by their original labeling and just work with these contracted edge trees. This reduces our model to the case where  for $i>1$, 
$T_i$ is a star tree with central node $a_i$, and $T_1$ has the form of two star trees, on $C$ and on $D$, that are joined at their central nodes by $e$. 

Let $P= P_1 + P'$ where $P_1$ is the mixture component from $T_1$, and $P'$ the sum of the components on the star trees $T_2= \dots = T_r$. 
Then one sees that
\[ 
M_2 := Flat_{A|B}(P') = N^T_A R N_B,
\]
with $R$ an $\kappa(r-1) \times \kappa(r-1)$ diagonal matrix giving the distribution of states at $a_i$ in components $2, \ldots, r$ weighted by $\pi_i$ and 
$N_A$, $N_B$ are stochastic matrices of sizes $\kappa(r-1) \times \kappa^{\# A}$, $\kappa(r-1) \times \kappa^{\# B}$ with entries giving conditional probabilities of states of variables in $A$, $B$ conditioned on states/components at the $a_i$. 
By choosing the positive root distributions at the nodes $a_i$, and positive $\pi_i$,  
$R$ is ensured to have positive diagonal entries, and hence have full rank. 

Consider $P_1$ where all matrices on pendant edges of $T_1$ are chosen to be $I_{\kappa}$. Also,  both the root distribution at $c$ and $M_e$ are chosen to have all positive entries. Then let
\[ 
M_1 := \Flat_{A|B}(P_1) = N^T_{1,A} R_1 N_{1,B}. 
\]
where $R_1$ is a $\kappa^2 \times \kappa^2$ 
diagonal matrix with entries giving the joint distribution at $c$ and $d$ weighted by $\pi_1$, and $N_{1,A}$, $N_{1,B}$ have all zero entries except for a single 1 in each row, and full row rank. 
Thus $M_1$ has rank $\kappa^2$. Moreover, it has at most one non-zero entry in each row and column, so both $\mathrm{im}(M_1)$ and $\ker(M_1)$ are coordinate subspaces. 

Since $\Flat_{A|B}(P)= M_1 + M_2$ and $\mu$ has the extended rank property $ERP(r,k)$, $\text{rank}(M_1+M_2)> \kappa r$ and thus we are done.
 \end{proof}

 Next goal is to show how to use the ranks of flattenings from Lemma \ref{lem:flatrank}
 to identify the key features of a phylogenetic mixture model.

\begin{lemma} \label{lem:find2splits}
   Suppose that $\mathbf{T} = (T_1, \ldots, T_r)  \in \mathcal{T}(r,n,k)$, and 
   let $P \in M^\mu_{\mathbf{T}}$ be generic.  Suppose that $\mu$ satisfies
   $RP(n,k)$ and $ERP(n,k)$.  Then we can use ranks of flattenings
   to find a partition on $[n]$ into three sets $A, B, C$ with $\#A \geq k$, $\#B \geq k$, such that $A|B \cup C$ and $B | A \cup C$ are valid splits for all 
   the trees $T_1, \ldots, T_r$.
\end{lemma}

\begin{proof}   
To find such an $A|B|C$, we simply test all the tripartitions that have $\#A \geq k$
and $\#B \geq k$.  We know, by the definition of $\mathcal{T}(r,n,k)$, that there must exist
a triple $A'|B'|C'$ where each of the splits $A'| B' \cup C'$, $B' |A' \cup C'$, and
$C'| A' \cup B'$ holds for all the trees $T_1, \ldots, T_r$.   
According to Lemma \ref{lem:flatrank},
a split $D|E$ has $\rank \Flat_{D|E}(P) \leq r\kappa$ if and only if
the split $D|E$ appears in all the trees $T_1, \ldots, T_r$, provided
that $\kappa^{\#D} > r \kappa$ and $\kappa^{\#E} > r \kappa$.  
By our assumptions on $k$, this holds.
So, the ranks of flattenings will find a triparition $A|B|C$ of the
desired type.  
\end{proof}

Note that just because we find $A|B \cup C$ and $B| A \cup C$
using the ranks of flattenings, it does not necessarily imply that 
$C|A \cup B$ is a valid split in all the trees.  This is in spite of
the fact that we know that there exists a triple $A'|B'|C'$ that
is a common tripartition to all trees.  Flattenings cannot necessarily
find that tripartition alone.  However, we can use Lemma \ref{lem:find2splits}
to prove our identifiability results anyways.  This result follows the
proof of Theorem 4.4 of \cite{Rhodes2012}.

\begin{lemma}   \label{lem:commontripart}
    Let $\mu$ be a standard Markov model that satisfies $RP(r,k)$ and
    the No Shuffling Property.  
    Suppose that the trees $\mathbf{T} = (T_1, \ldots, T_r)$ have
    a known common tripartition $A|B|C$ with $\#A \geq \#B \geq k$.  
    Then both $\mathbf{T}$ and the numerical parameters of the $\mu$-mixture model
    on $\mathbf{T}$ are generically identifiable.  
\end{lemma}

\begin{proof}
    Since all the trees in $\mathbf{T}$ share a common tripartition,
    we can write a distribution in the mixture model as a triple product
    \[
    \Flat_{A|B|C}(P) =  [M_A, M_B, M_C]
    \]
    where the matrices are $M_A$, $M_B$, and $M_C$ 
    are as described at the end of Section \ref{tensor}.  
    
    Since the model $\mu$ satisfies the rank property $RP(r,k)$,
    for generic choices of parameter the matrices $M_A$ and $M_B$
    will each have rank $r \kappa$.   The fact that $\mu$ is standard
    guarantees that $M_C$ has Kruskal rank $\geq 2$.  

    The triple product representation expresses $\Flat_{A|B|C}(P)$ 
    as the sum of $r \kappa$ rank $1$ tensors.  But since
    \[
    \rank_K(M_A) + \rank_{K}(M_B) + \rank_K(M_C)  \geq r \kappa + r \kappa + 2  = 2r\kappa + 2
    \]
    we can apply Kruskal's theorem to see that the matrices
    $M_A$, $M_B$, and $M_C$ can be recovered up to scaling and permuting the rows.

    Each of the rows of the recovered matrices $M_A$, $M_B$, $M_C$ will have entries
    from a scaled slice from a tree distribution on a subtree of one of the
    $T_i$ (the subtree from the common vertex to the leaves $A$).
    We need to group these rows together by the mixture components they come
    from.  However, since we have assumed that $\mu$ satisfies the
    No Shuffling Property, there is only one way to do this for $M_A$.
    Since the ordering of rows of $M_A$ denotes the order of the rows
    of $M_B$ and $M_C$, we can reassemble each scaled probability distribution
    $\pi_i P_i $ as the triple produce $[M_{i,A}, M_{i,B}, M_{i,C}]$
    for submatrices of $M_A, M_B, M_C$ respectively.

    From the scaled distribution $\pi_i P_i$ we recover the mixing weight via 
    the sum 
    \[
    \pi_i  =  \sum_{(j_1, \ldots, j_n) \in [\kappa]^n}  \pi_i P_i(j_1, \ldots, j_n).
    \]
    Then we can get the distribution $P_i$ for the single tree $T_i$.
    We use that $\mu$ is a standard model, so  the tree parameter $T_i$ and numerical
    parameters are generically identifiable.
\end{proof}

The next Lemma closely follows the proof of Theorem 4.6 in \cite{Rhodes2012}.

\begin{lemma}\label{lem:findfrom2splits}
    Let $\mu$ be a standard Markov model that satisfies $RP(r,k)$ and
    the No Shuffling Property.  
    Suppose that the trees $\mathbf{T} = (T_1, \ldots, T_r)$ have
    a known tripartition with $A|B|C$ with $\#A \geq \#B \geq k$ such that
    $A|B \cup C$ and $B | A \cup C$ are valid splits in all trees $T_1, \ldots, T_r$.
        Then both $\mathbf{T}$ and the numerical parameters of the $\mu$-mixture model
    on $\mathbf{T}$ are generically identifiable.  
\end{lemma}

\begin{proof}
    Let $P$ be a generic distribution of the $\mu$ mixture model.
    Fix some $c \in C$, let $D(c) = A \cup B \cup \{c\}$.
    Consider the marginalization of $P$ to the set $D(c)$, and call this
    distribution $P_c$. This is the probability tensor for the induced $r$-tuple of trees 
    \[
    \mathbf{T}|_{D(c)}  = (T_1|_{D(c)}, \ldots, T_r|_{D(c)}).  
    \]
    Note that all the trees of $\mathbf{T}|_{D(c)} $ share the common triparition
    $A|B|\{c\}$, which satisfies the conditions of Lemma \ref{lem:commontripart},
    so that tree parameters and numerical parameters are generically
    identifiable on this subset of the leaves.  We also have that
    \[
    \Flat_{A|B|\{c\} }(P_c)  =  [M_A, M_B, M_c].
    \]
    Since $\mu$ satisfies the rank property $RP(r,k)$, we can assume that
    $M_A$ has full rank.  Thus, there is a right inverse matrix $Q_A$
    with the property that $M_A Q_A = I_{r \kappa}$, where $I_{r \kappa}$
    is the $r \kappa$ identity matrix.  Our goal is to use the matrix $Q_A$
    to finish the disentangling of distributions that go into $P$.

    Now, since $A|B \cup C$ is a valid split in all the trees $\mathbf{T}$,
    we can write a factorization
    \[
    \Flat_{B \cup C | A}(P)  = M^T_{B \cup C} \Pi \widetilde{M_A}
    \]
    where 
    $\tilde{M_A}$  and $M_{B \cup C}$ are stochastic matrices of probabilities
    of the states of the leaves in $A$ and $B\cup C$ conditioned which tree $T_i$
    we are in and on the states at the root $w_i$ in each tree where $T_i|_A$ attaches
    to the rest of $T_i$.  The matrix $\Pi$ is a diagonal matrix
    whose entries are the root distribution probabilities times the mixing weights.  
    A key feature is that the order of all components can be taken so that
    all the parameters associated to a particular tree $T_i$ can be assumed
    to be in the same block of rows.  

    Now we can write $M_A  =  R\widetilde{M_A} $ where $R$ is a block diagonal
    matrix whose $i$th block gives the conditional probability
    of state changes from the root $w_i$ to the adjacent vertex in $T_i$
    away from $A$.  We can assume that this matrix is invertible since
    the blocks are just $\mu$-transition matrices, and generically those are
    nonsingular since $\mu$ is standard.  

    Now we can compute
    \[
    \Flat_{B\cup C|A}(P) Q_A  =  M^T_{B \cup C} \Pi R^{-1} M_A Q_A  =  M^T_{B\cup C} \Pi R^{-1}.  
    \]
    This shows that taking the columns of $\Flat_{B \cup C | A}(P) Q_A$ in blocks of 
    $\kappa$ we obtain entries associated to only one mixture component at
    a time.  Multiplying a block of those columns by the associated corresponding
    rows of $M_A = R \widetilde{M_A}$ be obtain a single mixture component
    $\pi_i  P_i$ from the single tree $T_i$, multiplied by the mixing weight
    $\pi_i$.   We can identify $\pi_i$ by summing all entries of this tensor,
    as in the proof of Lemma \ref{lem:commontripart}.  Then the fact that
    the tree and numerical parameters are generically identified for the model $\mu$
    for a single tree completes the proof.  
\end{proof}

Now we are in a position to combine all components to complete the
proof of the main structural theory on identifiability of mixture models.

\begin{proof}[Proof of Theorem \ref{thm:main}]
    Let $P$ be a generic distribution from a mixture model with
    some $\mathbf{T} = (T_1, \ldots, T_r) \in \mathcal{T}( r,n,k)$.  
    We need to show that from $P$ alone, we can find  $\mathbf{T}$
    and the numerical parameters of the model.  According to Lemma \ref{lem:find2splits},
    it is possible to use ranks of flattenings to find a tripartiion $A|B|C$ of the
    leaf set such that $A|B \cup C$ and $B|A \cup C$ are valid splits in all
    the trees $T_i$, and both $\#A \geq k$ and $\#B \geq k$.  
    Once those are identified, Lemma \ref{lem:findfrom2splits}
    shows that it is possible to identify the trees $T_i$, the mixing
    weights, and the numerical parameters for each tree.    
\end{proof}


\section{Rank property, Extended Rank Property, and Standard Property
for the Jukes-Cantor Model}  \label{sec:rankproperty}

In this section, we will prove that the Jukes-Cantor
model satisfies the Rank Property, Extended Rank Property,
and Standard condition with appropriate conditions
on $r$ and $k$.  While this might seem narrow,
a key observation is that if any of these properties are satisfied
for a certain model $\mu$, they are also satisfied for all
models $\mu'$ that contain $\mu$ as submodels.  Since the
Jukes-Cantor model is contained in all the equivariant models
as a submodel, this will prove that those three properties 
are also satisfied for those models.

First we will prove that the models under consideration
in this paper are all standard models, the most straightforward 
property to prove for a model.

\begin{lemma}\label{lem:JCnontrivial}
    The JC, K2P, K3P, SSM, and GMM models are all standard Markov models.
\end{lemma}

\begin{proof}
    All five models contain the identity matrix, and are known to be
    identifiable on trees 
    (e.g. using the tensor rank arguments from \cite{Allman2011}).
    
    Recall that the Jukes-Cantor model consists of all transition
    matrices that have one value $b$ for all off diagonal entries
    and a different value $a$ for all diagonal entries.  
    Clearly, the Jukes-Cantor model contains the identity matrix setting
    $b = 0$ and $a = 1$.  For the second property related to Kruskal
    ranks we can take the identity matrix, together with any other
    matrix in the model that does not have rank $1$  (e.g.~ take
    $b = \epsilon$ and $a = 1 - (\kappa - 1) \epsilon$).  
    This pair of matrices will satisfy the condition on independence
    of column vectors.
\end{proof}

A key fact that we will use (and that was also used in the proofs in 
\cite{Rhodes2012}), is that matrices that are generalized Vandermonde matrices
have full rank for generic choices of parameters.

\begin{definition}
Let $x^{u_1}, \ldots, x^{u_r}$  be  monomials in $\cc[x_1, \ldots, x_n]$.  
Let $v_1, \ldots, v_s \in \cc^n$ vectors.  Then the matrix 
\[
V(u_1, \ldots, u_r; v_1, \ldots, v_s)  =
\begin{pmatrix}
    v_1^{u_1} &  v_1^{u_2} & \cdots  & v_1^{u_r} \\
    \vdots  & \vdots  & \ddots  & \vdots  \\
    v_s^{u_1} & v_s^{u_2} & \ddots & v_s^{u_r} 
\end{pmatrix}
\]
is
called the \emph{generalized Vandermonde matrix}.
\end{definition}

\begin{prop}\label{prop:vandermonde}
    Let $x^{u_1}, \ldots, x^{u_r}$  be  distinct monomials in $\cc[x_1, \ldots, x_n]$.  
Let $v_1, \ldots, v_r \in \cc^n$ generic vectors.  Then the generalized
Vandermonde matrix $V(u_1, \ldots, u_r; v_1, \ldots, v_r)$ has rank $r$.
\end{prop}

\begin{proof}   
Since all vectors $u_1, \ldots, u_r$ are distinct, there are integers is an assignment of 
variables $x_i = t^{d_i}$, so that the resulting powers of the variable $t$
$t^{u_1 d}$, $t^{u_2 d}$, $\ldots$ are all distinct.  Then the result follows
from the fact that the standard Vandermonde matrix has full rank for generic values.
\end{proof}

More generally, we have the following useful fact about the matroid that
is determined by a generalized Vandermonde matrix.  

\begin{prop}\label{lem:vanderzeros}
Let $x^{u_1}, \ldots, x^{u_r}$  be  not necessarily 
distinct monomials in $\cc[x_1, \ldots, x_n]$.  Suppose that there are $l$
distinct monomials among them.
Let $v_1, \ldots, v_s \in \cc^n$ generic vectors.     Then
any nonzero vector in the row span of 
$V(u_1, \ldots, u_r; v_1, \ldots, v_s)$ has at least  $l - s+1$ nonzero entries.  
\end{prop}

\begin{proof}
    We can assume that there are no repeats in the list of monomials, since
    any repeats must necessarily yield repeated nonzero entries of vectors in the
    row span. 
    Suppose that there is a nonzero vector $x$ in the row span of $V = V(u_1, \ldots, u_r; v_1, \ldots, v_s)$ that had $s$ or more nonzero entries.  Let $S$ be a set
    of exactly $s$ of those entries.  After permuting columns, we can assume
    that those are the first $s$ entries of $x$.  The fact that $x$ is a nonzero
    vector of the row span means that there is an invertible $s \times s$ matrix $M$
    such that $MV$ has $x$ as the bottom row.   Consider the square-submatrix
     $V' =  V(u_1, \ldots, u_s; v_1, \ldots, v_s)$ obtained by taking the
     first $s$ columns of $V$.  By Proposition \ref{prop:vandermonde}
     $V'$ is invertible,  so $MV'$ is also invertible.   But $MV'$ has a row of all zeroes (from taking the subvector of $x$).  
     This shows that there must be at least $r - s + 1$ non-zero entries in $x$.
\end{proof}

\begin{lemma}\label{lem:rpJC}
    The Jukes-Cantor model satisfies the rank property $RP(r,k)$
    for $r \leq 2^{k-1} - k + 1$ when $\kappa \geq 3.$ 
\end{lemma}

\begin{proof}
  Let $r$ and $k$  satisfy $r \leq 2^k - k+1$.  We must show that
  for any trees $T_1, \ldots, T_r$ on $k$ leaves, and generic probability
distributions $P_i \in M^\mu_{T_i}$, the matrix
\[
\begin{pmatrix} \Flat_{[n-1]|n}(P_1) & \Flat_{[n-1]|n}(P_2)  & \cdots & \Flat_{[n-1]|n}(P_r)
\end{pmatrix}
\]
has rank $r \kappa$.
However, by the fact that rank of a matrix being $\leq \alpha$ is a closed condition,
it suffices to show that there is a single choice of parameters that gives the rank
$r \kappa$.  

Suppose that $T$ is a tree.  
Note that if we set the transition matrix on an edge to be the identity
matrix (which is a transition matrix in the Jukes-Cantor model), that will give
a probability distribution on a tree obtained from $T$ by contracting the
corresponding edge.  Hence, if we set the transition matrices of
all internal edges of the tree to be the identity matrix, this will give
us distributions on the star tree.  So distributions on the star
tree appear as distributions in the model on any tree.
Thus the result follows if we prove the Lemma when $T_1 = T_2 = \cdots = T_r$
are all the star trees.

To this end, let $T_1, \ldots, T_r$ be star trees with $k$ leaves,
and consider the matrix 
\[
M = \begin{pmatrix} \Flat_{[k-1]|k}(P_1) & \Flat_{[k-1]|k}(P_2)  
   & \cdots & \Flat_{[k-1]|k}(P_r)
\end{pmatrix}.
\]
Our first step is to apply the Fourier transform to the parametrization.  Let
$Q_i$ be the Fourier transformation of the probability distribution $P_i$.  
 The Fourier transform is linear, and
it transforms the matrix $M$ into a new matrix $M'$ which is
\[
M' = \begin{pmatrix} \Flat_{[k-1]|k}(Q_1) & \Flat_{[k-1]|k}(Q_2)  & \cdots & 
\Flat_{[k-1]|k}(Q_r)
\end{pmatrix}.
\]
The new matrix $M'$ is obtained from $M$ by row and column operations,
so $M$ and $M'$ have the same rank.

Now we analyze the parameterization in the Fourier coordinates.  For the Jukes-Cantor model on a star tree $T_i$ we have that
\[
q(g_1, \ldots, g_k)   = 
\begin{cases}  \prod_{i = j}^k a^{(i, j)}_{g_j}  &  \sum_{j = 1}^k g_j  = 0  \\
0  &   \mbox{otherwise}
\end{cases}
\]
where we have a set of parameter $a^{(i,j)}_g$ for each tree $T_i$, each edge $j$,
and each group element $g \in G$.  For the Jukes-Cantor model
we have that 
\[
a^{(i,j)}_{(1,0)} = a^{(i,j)}_{(0,1)} = a^{(i,j)}_{(1,1)} 
\]
for all $i$ and $j$.  We can rearrange rows and columns of $M'$ so that
is has a block form, grouping all the columns by the value of  $g_k$,
and grouping the rows so that, for a fixed value of $g_k$
we have all the $(g_1, \ldots, g_{k-1})$ so that $q(g_1, \ldots, g_k)  \neq 0$
together.  

After this rearrangement of rows and columns, $M'$ will be a block diagonal
matrix, with $\kappa$ blocks, each block of size $\kappa^{k-2} \times r$.
We need to show that each of these blocks has full rank, so we get that
the total rank is $r \times \kappa$ as desired.  

To show that the block matrices have the appropriate rank, we note that
each such matrix is a generalized Vandermonde matrix.  Indeed,
each entry of the matrix is a monomial, and each column is an identical copy 
of the first column, but with new variables.  Note, however,
that the condition $a^{(i,j)}_{(1,0)} = a^{(i,j)}_{(0,1)} = a^{(i,j)}_{(1,1)}$
will yield repeated monomials in each column.  Thus to complete
the proof, we need to figure out how many distinct monomials
there are, so we can apply Proposition \ref{prop:vandermonde}.

Consider the map from 
\[
\phi:  \zz_2 \times \zz_2 \rightarrow \{0,1\},  
\phi(g,h)  =   \begin{cases}
    0  &  \mbox{ if }  g = h = 0  \\
    1  &  \mbox{ otherwise}.
\end{cases}
\]
Then two monomials $q(g_1, \ldots, g_k)$  and $q(h_1, \ldots, h_k)$
are identical if and only if  $\phi(g_i) =  \phi(h_i)$ for all $i$.
So we just need to count the number of equivalence classes for each
fixed value of $g_k$.

This is straightforward to do:  if $\phi(g_k) = 0$, then
$(\phi(g_1), \ldots, \phi(g_{k-1}))$ can be any string in $\{0,1\}^{k-1}$
except the strings that have exactly one $1$.  There are $2^{k-1} - k + 1$
such strings.  On the other hand, if $\phi(g_k) \neq 0$ then $(\phi(g_1), \ldots, \phi(g_{k-1}))$ can be any string in $\{0,1\}^{k-1}$ except the string with all zeroes.
There are $2^{k-1} - 1$ such strings.  We need to take the smaller of these
two values to get a consistent rank across all the blocks.
Hence this shows that the Jukes-Cantor model satisfies
$RP(r,k)$ with $r \leq 2^{k-1} - k + 1$.
\end{proof}

\begin{corollary}
    The Kimura 2-parameter model, the Kimura 3-parameter model, 
    and the strand symmetric model all satisfy the rank property
    $RP(r,k)$ with $r \leq 2^{k-1} - k + 1$.
\end{corollary}

\begin{proof}
    We just need to show the existence of a single choice of parameters
    that give the desired rank condition.  However, since the Jukes-Cantor
    model is a submodel of all of those other models, the result of
    Lemma \ref{lem:rpJC} gives the desired result.
\end{proof}

Note that the bound $r \leq 2^{k-1} - k + 1$ is not best
possible for those other models besides the Jukes-Cantor model. Each model would require a more careful analysis to improve the results.
Following the proof, for the group-based models,
it suffices to determine the number of distinct monomials of different
types in the block structure of the matrix $M'$.  Both Kimura 2-parameter and
3-parameter models will have significantly more distinct monomials
than the Jukes-Cantor model, and so the rank property will hold for
larger values of $r$.

Now we proceed to prove the Extended rank property for the Jukes-Cantor model.
Again, that will also give a result for other models, though it is probably not
the best possible for K2P, K3P, SSM.

\begin{lemma}\label{lem:erpJC}
    The Jukes-Cantor models satisfies the extended rank property $ERP(r,k)$
    for $r \leq 2^{k-1} - k+ 1$ when $\kappa \geq 3$. 
\end{lemma}

\begin{proof}
Let $T_1$ be a tree on  $k+1$ leaves that does not contain the split $[k-1]| \{k, k+1\}$,
and any trees $T_2, \ldots T_r$  on $k$ leaves, and generic probability
distributions $P_i \in M^\mu_{T_i}$.  Consider the matrix the matrix
\[
M = 
\begin{pmatrix} \Flat_{[k-1]|\{k, k+1\}}(P_1) & \Flat_{[k-1]|k}(P_2)  & \cdots & \Flat_{[k-1]|k}(P_r)
\end{pmatrix}.
\]
We must show that $M$
has rank $(r-1) \kappa + \kappa^2$, generically.

    As in the proof of Lemma \ref{lem:rpJC}, it suffices to prove that there
    is a single choice of parameters that achieves the desired rank.
    Then we can set many parameters equal to identity matrices, and
    consider the resulting trees that arise by contracting those edges.
    To that end, we can assume that trees $T_2, \ldots, T_r$ are all
    $k$ leaf star trees.  
    
    As in the proof of Lemma \ref{lem:rpJC}, we apply the Fourier transform to all
    probability distributions to get a matrix 
\[
M' = 
\begin{pmatrix} \Flat_{[k-1]|\{k, k+1\}}(Q_1) & \Flat_{[k-1]|k}(Q_2)  & \cdots & \Flat_{[k-1]|k}(Q_r)
\end{pmatrix}.
\]
Looking at the final blocks we have the matrix
\[
\begin{pmatrix} \Flat_{[k-1]|k}(Q_2)  & \cdots & \Flat_{[k-1]|k}(Q_r)
\end{pmatrix}.
\]
This is the same matrix we have seen in the 
proof of Lemma \ref{lem:rpJC}, with one fewer set of columns.
So it  has rank $(r-1) \kappa$ since JC has the rank property $RP(r-1,k)$ with
these values of $r$ and $k$.  
Furthermore, after reordering rows and columns, as in the proof of Lemma \ref{lem:rpJC},
it can be broken into $4$  blocks, each of which is a generalized Vandermonde
matrix.  

The $4$ generalized Vandermonde matrices from the previous paragraph each have
size $r-1  \times 4^{k-2}$, and those generalized Vandermonde matrices each have
full rank.  We want to use Lemma \ref{lem:vanderzeros}
to complete the proof.  In particular, we will show that there is a choice
of parameters for the tree $T_1$ so that the resulting 
matrix $\Flat_{[k-1]|\{k, k+1\}}(Q_1)$
has the property that each of the four vectors that it contributes to
the support of one of the four Vandermonde submatrices has 
at most 4 nonzero entries.  This will insure that
$M'$ has the correct rank, by applying Lemma \ref{lem:vanderzeros}.
With these observations in mind, we consider various restriction
on the tree $T_1$ and matrices $\Flat_{[k-1]|\{k, k+1\}}(Q_1) $
that can be produced.

First of all, $T_1$ must have a split that is
    that is not compatible with $[k-1] | \{k, k+1\}$.  
    After relabeling the leaves, we can assume this split has the form
    \[
    A_j | B_j = \{1, \ldots, j-1, k\}|  \{j, \ldots, k-1, k+1\}  
    \]
    for some $j$ between $2$
    and $k -1$.
 We can assume that $T_1$ is the tree with only this
    one split $A_j|B_j$ as an internal split, by setting
    all other internal edges to have an identity matrix
    as the transition matrix.   We must consider a few scenarios
    based on the sizes of the sets $A_j$ and $B_j$.

{\bf Case 1:}  Both $\#A_j$ and $\#B_j$ are even.  In this case,
we take all the transition matrices associated to pendant edges
in the tree to be the identity matrix as well.  
Our assumption on $T_1$ yields a specific structure on the Fourier coordinates
    of the distributions in the model on $T_1$.  
    Considering a vector $(g_1, \ldots, g_{k+1}) \in G^{k+1}$
    we can break this into blocks $(\mathbf{h}_1, \mathbf{h}_2, g_{k}, g_{k+1} )$,
where
\[
\mathbf{h}_1  =  (g_1, \ldots, g_{j-1})  \quad  \mathbf{h}_{2} = (g_{j}, \ldots, g_{k-1}).
\]
    For a Fourier coordinate $q(\mathbf{h}_1, \mathbf{h}_2, g_{k}, g_{k+1} )$
    to be nonzero, we must have the following
    \begin{enumerate}
        \item $g_1 = g_2 = \cdots = g_{j-1} = g_{k}$
        \item  $g_j = g_{j+1} = \cdots = g_{k-1} = g_{k+1}$
        \item  $\sum_{j = 1}^{k+1} g_j = (0,0)$.
    \end{enumerate}
    Conditions (1) and (2) are coming from the fact that the only nontrivial 
    transition matrix for the model is the one corresponding to the
    split $A_j|B_j$.  Condition (3)  is the standard condition
    for Fourier coordinates in a group-based model.
    However, since both $\#A_j$ and $\#B_j$ are even, condition (3) is automatically
    satisfied by any vectors that satisfy conditions (1) and (2).  
    Note that there are
    exactly $16$ nonzero values of 
    $q(\mathbf{h}_1, \mathbf{h}_2, g_{k}, g_{k+1} )$,  one for each
    of the possible pairs $(g_k, g_{k+1}) \in G^2$.  
    Since the columns of $\Flat_{[k-1]|\{k, k+1\}}(Q_1) $
    are indexed by those pairs, we see that each column of $\Flat_{[k-1]|\{k, k+1\}}(Q_1) $
    has exactly one nonzero entry, and they appear in different rows.
    Furthermore, the rearrangement of rows and columns so that 
    \[
\begin{pmatrix} \Flat_{[k-1]|k}(Q_2)  & \cdots & \Flat_{[k-1]|k}(Q_r)
\end{pmatrix}
\]
is a block matrix has blocks indexed by the value of $g_k$.
Hence, we have that each block of $\Flat_{[k-1]|\{k, k+1\}}(Q_1) $
contributes a four dimensional column space, all of which have at most
$4$ nonzero entries.  This shows that $M'$ will have rank $(r-1)4 + 4^2$, as desired.

{\bf Case 2:}  One of $\#A_j$ and $\#B_j$ is even, and one is odd.  
We can assume that $\#A_j$ is odd and $\#B_j$ is even.  
We set all the pendant edge parameters to the identity matrix
except for the edge going to leaf $k+1$.  As in Case 1,
 we consider a vector $(g_1, \ldots, g_{k+1}) \in G^{k+1}$
    we can break this into blocks $(\mathbf{h}_1, \mathbf{h}_2, g_{k}, g_{k+1} )$,
where
\[
\mathbf{h}_1  =  (g_1, \ldots, g_{j-1})  \quad  \mathbf{h}_{2} = (g_{j}, \ldots, g_{k-1}).
\]
    For a Fourier coordinate $q(\mathbf{h}_1, \mathbf{h}_2, g_{k}, g_{k+1} )$
    to be nonzero, we must have the following
\begin{enumerate}
        \item $g_1 = g_2 = \cdots = g_{j-1} = g_{k}$
        \item  $g_j = g_{j+1} = \cdots = g_{k-1}$
        \item  $\sum_{j = 1}^{k+1} g_j = (0,0)$.
    \end{enumerate}
Note the change that we will not need $g_{k+1}$ to be equal to the other
values in Condition (2).  Hence we have three groups of coordinates, each
of which have an odd number of elements (the groups being $A_j$,  $\{k+1\}$ and
$B_j \setminus \{k+1\}$).  Let
\[
h_1 := g_1 = g_2 = \cdots = g_{j-1} = g_{k}  \quad 
\mbox{and} \quad
h_2 := g_j = g_{j+1} = \cdots = g_{k-1}
\]
then we get a valid coordinate when
\[
h_1 + h_2 + g_{k+1} = (0,0).
\]
This follows because each $h_i$ is equal to the sum of the $g_l$'s in its group, because all are equal and the number of such elements is odd.
We see that there are exactly $16$ possibly solutions (choosing values for $h_1$ and $h_2$
arbitrarily forces a value for $g_{k+1}$.  Furthermore, all possible pairs coordinates
$(g_{k}, g_{k+1})$ are possible.  Thus, as in Case 1, we see that each column of $\Flat_{[k-1]|\{k, k+1\}}(Q_1) $
    has exactly one nonzero entry, and they appear in different rows.
    Furthermore, the rearrangement of rows and columns so that 
    \[
\begin{pmatrix} \Flat_{[k-1]|k}(Q_2)  & \cdots & \Flat_{[k-1]|k}(Q_r)
\end{pmatrix}
\]
is a block matrix has blocks indexed by the value of $g_k$.
Hence, we have that each block of $\Flat_{[k-1]|\{k, k+1\}}(Q_1) $
contributes a four dimensional column space, all of which have at most
$4$ nonzero entries.  This shows that $M'$ will have rank $(r-1)4 + 4^2$, as desired.

{\bf Case 3:}  Both $\#A_j$ and $\#B_j$ are odd.
We set all the parameters corresponding to pendant edges
to the identity matrix except for leaf 1, which we allow to
be arbitrary.  As in Case 1,
 we consider a vector $(g_1, \ldots, g_{k+1}) \in G^{k+1}$,
 and we analyze which Fourier coordinates
 $q(g_1, \ldots, g_{k+1}) $ can be nonzero.  With our assumption
 on the transition matrices, we must have
     \begin{enumerate}
        \item $ g_2 = \cdots = g_{j-1} = g_{k}$
        \item  $g_j = g_{j+1} = \cdots = g_{k-1} = g_{k+1}$
        \item  $\sum_{j = 1}^{k+1} g_j = (0,0)$.
    \end{enumerate}
Note that $g_1$ does not appear in the first group.  So we have three groups
where we will have equal values
\[
h_1 := g_2 = \cdots = g_{j-1} = g_{k}  \quad \mbox{and} \quad 
h_2 := g_j = g_{j+1} = \cdots = g_{k-1} = g_{k+1}.
\] 
Since the first group $A_j \setminus \{1\}$ has an even number of elements,
the sum of these elements will always be $(0,0)$ regardless of what $h_1$ is.
On the other hand, the second group has an odd number of elements, so $h_2$ equals
the sum of all those elements.  Then by condition (3), we are forced to have
$g_1 = h_2$.  Hence, there are 16 possible Fourier coordinates that
have nonzero entries.  As in Cases 1 and 2, they allow for all 16 different
possibilities for the pairs $(g_{k}, g_{k+1})$, and will hence give
that $M'$ has rank $(r-1)4 + 4^2$ as desired.  
\end{proof}

As in the case of the rank property, we also can see a similar result
for the extended rank property, for the other equivariant models
we have studied.  We omit the proof which is the same as for the
rank property. Again, these results are probably not best possible for
those other models, and paying attention to the structure of the
generalized Vandermonde matrices that arise in the other models
can yield stronger results.

\begin{corollary}
    The Kimura 2-parameter model, the Kimura 3-parameter model, 
    and the strand symmetric model all satisfy the extended rank property
    $ERP(r,k)$ with $r \leq 2^{k-1} - k + 1$.
\end{corollary}


\section{The No Shuffling Property}  \label{sec:noshuffle}

In this section, we prove that the models JC, K2P, K3P, and SSM
satisfy the No Shuffling Property.
Like the results from the previous section on the rank property
and extended rank property, we are able to just show the result
for the JC model and immediately deduce it for all other models.

Recall the No Shuffling Property:

\begin{definition}
A Markov model $\mu$ is said to the have the \emph{No Shuffling Property} 
if the following holds:
For any trees $T_1, \ldots, T_r$ with $n \geq 3$ leaves 
and generic probability distributions
$P_i \in M^\mu_{T_i}$, let $Q$ be the matrix
\[
Q = \begin{pmatrix} \pi_1 \Flat_{[n-1]|n}(P_1) & \pi_2 \Flat_{[n-1]|n}(P_2)  & \cdots & \pi_r\Flat_{[n-1]|n}(P_r)
\end{pmatrix}
\]
where $\pi \in \Delta_r$ be generic.  
Form a new matrix $Q'$ by taking $\kappa$ columns of $Q$, and let $P$
be the resulting tensor.  If $P \in M^\mu_T$ for some $T$, then all
columns of $Q'$ must have come from the same $P_i.$
\end{definition} 

Note the No Shuffling Property does not have any conditions on $r$.  The results
should hold for every $r$ (and we can of course, stop at $r = \kappa$).
Also, note that if we prove the No Shuffling Property for a Markov model $\mu$
just for trees with $3$ leaves, this will prove that the property
also holds for trees with $n$ leaves with $n \geq 3$ as well.
Otherwise, there is some $n$, and trees $T_1, \ldots, T_r$ and generic distributions
$P_i \in M^\mu_{T_i}$, where we form the matrix 
\[
Q = \begin{pmatrix} \pi_1 \Flat_{[n-1]|n}(P_1) & \pi_2 \Flat_{[n-1]|n}(P_2)  & \cdots & \pi_r\Flat_{[n-1]|n}(P_r)
\end{pmatrix}
\]
where $\pi \in \Delta_r$ be generic.  Then we can 
form a new matrix $Q'$ by taking $\kappa$ columns of $Q$, and let $P$
be the resulting tensor.  If $P \in M^\mu_T$ for some $T$, then
we can marginalize to tree leaves and the same statement will be true for
three leaves as well, contradicting that we had the No Shuffling Property
for $n = 3$.  Hence, for the rest of this section, we will restrict to
three leaf trees.

To show the No Shuffling Property, we will use phylogenetic invariants
that come from the general Markov model, specifically certain 
commutation invariants described in  \cite{Allman2003, Strassen1983}. 

\begin{theorem}\label{thm:tensorinve}
    Let $P$ be a $\kappa \times \kappa \times \kappa$ tensor
    giving a distribution from the general Markov model with $\kappa$ states
    on a $3$ leaf tree with $\kappa > 2$.  For $i = 1, \ldots, \kappa $
    let $P_{(i)}$ be the $\kappa \times \kappa$ matrix slice $P_{(i)} = ( P(i,u,v))_{u,v}$.
    Then
    \[
    P_{(i)}  (\mathrm{adj}(P_{(j)}) ) P_{(k)} - P_{(k)}  (\mathrm{adj}(P_{(j)}) ) P_{(i)}
    \]
    for any $i,j,k$.
\end{theorem}

 Note that $\mathrm{adj}(A)$ denotes the classical adjoint of $A$, which is $\det(A) A^{-1}$ for an invertible matrix, and is well-defined even for singular matrices
 (since the entries are all polynomials in the entries of $A$).

 \begin{lemma}\label{lem:noshuff}
     Let $\mu$ be any phylogenetic Markov model on $\kappa = 4$ states that
     contains the Jukes-Cantor model.  Then $\mu$ satisfies the No Shuffling Property.
 \end{lemma}

\begin{proof}
  We proved this result via a computation, and using the phylogenetic invariants
  from the general Markov model described in Theorem \ref{thm:tensorinve}.
  Specifically, we generated three random tensors from the Jukes-Cantor model
  parameterization  $P^1, P^2, P^3$.     Then from those $ 4\times 4 \times 4$ tensors,
  we choose three matrix slices, and evaluate them in the polynomials from 
  Theorem \ref{thm:tensorinve}.  We find that that this evaluates to zero only if
  the three matrix slices come from the same $P^i$.   

  The fact that the invariants from the General Markov Model evaluate to 
  zero only when all three slices come from the same $P^i$ proves that all
  models that contain the Jukes-Cantor model will satisfy the No Shuffling Property.
  Indeed, if the No Shuffling Property is not satisfied for some model
  $\mu$, then for all generic distributions from that model,
  there will be some choice of slices from different $P^1, P^2, P^3$ that satisfy the invariants that come
  from that model.
  The fact that the invariants vanish for the generic distributions in the model,
  implies that they vanish for all distributions that come from the model.
  Since the General Markov Model contains every phylogenetic model,
  the invariants for the General Markov model are contained in the 
  ideal of invariants for any other model.  Since we have produced a distribution
  in the Jukes-Cantor model that does not satisfy the phylogenetic invariants
  for the General Markov model with any non-standard set of slices,
  this proves that no model that contains the Jukes-Cantor model will satisfy
  those invariants for General Markov model.  Hence any model containing
  the Jukes-Cantor model satisfies the No Shuffling Property.     
\end{proof}

With the tools from Sections \ref{sec:rankproperty} and \ref{sec:noshuffle} in
hand, we are ready to apply Theorem \ref{thm:main} to deduce Corollary \ref{cor:JCetc}.

\begin{proof}[Proof of Corollary \ref{cor:JCetc}]
According to Theorem \ref{thm:main}, to deduce Corollary \ref{cor:JCetc}, 
we must show that all of the models JC, K2P, J3P, SSM, GMM:
\begin{enumerate}
\item  are standard Markov models,
\item  satisfy the rank property $RP(2^k-k, k)$, 
\item  satisfy the extended rank property $ERP(2^k-k, k)$,  and
\item  satisfy the No Shuffling Property.
\end{enumerate}

The fact that all the models are standard  is the content of Lemma \ref{lem:JCnontrivial}.

The JC model satisfies the rank property $RP(2^k-k, k)$ by Lemma \ref{lem:rpJC}.
Any supermodel of the JC model will hence satisfy $RP(2^k-k, k)$ as well.
The JC model satisfies the extended rank property $ERP(2^k-k, k)$
by Lemma \ref{lem:erpJC}.  Any supermodel of the JC model will hence satisfy 
$ERP(2^k-k, k)$ as well.  Finally, Lemma \ref{lem:noshuff} shows
that all the models that contain JC satisfy the No Shuffling Property.    
\end{proof}

\section*{Acknowledgments}
This material is based in part upon work supported by the National Science Foundation under Grant No.~DMS-1929284 while the authors were in residence at the Institute for Computational and Experimental Research in Mathematics in Providence, RI, during the semester program on
``Theory, Methods, and Applications of Quantitative Phylogenomics''.

\bibliographystyle{plain}
\bibliography{JC}

\end{document}